\documentclass[conference]{IEEEtran}
\normalsize

\usepackage[T1]{fontenc}
\usepackage{amsmath,amssymb,amsfonts}
\usepackage{mathrsfs}
\usepackage{amsbsy}
\usepackage{graphicx}
\usepackage{graphics}
\usepackage{epstopdf}
\usepackage{algorithm}
\usepackage{algorithmic}
\usepackage{subfigure}
\usepackage{cite}
\usepackage{slashbox}
\usepackage{multirow}
\usepackage{ctable}
\usepackage{tikz,pgfplots}

\newtheorem{theorem}{{Theorem}}
\newtheorem{lemma}[theorem]{{Lemma}}

\DeclareMathAlphabet{\mathbfsl}{OT1}{ppl}{b}{it} 

\newcommand{\be}[1]{\begin{equation}\label{#1}}
\newcommand{\ee}{\end{equation}}
\newcommand{\eq}[1]{(\ref{#1})}

\renewcommand{\leq}{\leqslant}

\newcommand{\Tref}[1]{Theo\-rem\,\ref{#1}}

\newcommand{\Lref}[1]{Lem\-ma\,\ref{#1}}
\newcommand{\Cref}[1]{Co\-ro\-lla\-ry\,\ref{#1}}

\newcommand{\eps}{{\epsilon}}
\newcommand{\Pef}{P_{e,\text{frame}}}
\newcommand{\F}{\mathbb{F}}

\begin{document}

\title{On the  Construction and Decoding of \\ Concatenated Polar Codes}
%
%
%

\author{%
\authorblockN{\large{
Hessam Mahdavifar, 
        Mostafa El-Khamy, 
        Jungwon Lee,
        Inyup Kang}}\\
\authorblockA{
Mobile Solutions Lab, Samsung Information Systems America\\
4921 Directors Place, San Diego, CA 92121\\
{ \{h.mahdavifar,\,mostafa.e,\,jungwon2.lee,\,inyup.kang\}@samsung.com}}\vspace*{-1.5ex}
}


%
%

\maketitle

\begin{abstract}
A scheme for concatenating the recently invented polar codes with interleaved block codes is considered.
By concatenating binary polar codes with interleaved Reed-Solomon codes, we prove that the proposed concatenation scheme captures the capacity-achieving property of polar codes, while having a significantly better error-decay rate. We show that for any $\epsilon > 0$, and total frame length $N$,  the parameters of the scheme can be set such that the frame error probability is less than $2^{-N^{1-\epsilon}}$, while the scheme is still capacity achieving. This improves upon $2^{-N^{0.5-\eps}}$, the frame error probability of Arikan's polar codes. We also propose decoding algorithms for concatenated polar codes, which significantly improve the error-rate performance at finite block lengths while preserving the low decoding complexity.
\end{abstract}


%

\section{Introduction}
\label{sec:Introduction}
Polar codes, introduced by Arikan \cite{Arikan, AT}, is the most recent breakthrough in coding theory. Polar codes are the first and, currently, the only family of codes with explicit construction (no ensemble to pick from) to achieve the capacity of a certain family of channels (binary input symmetric discrete memory-less channels) as the block length goes to infinity. They have encoding and decoding algorithms with very low complexity. Their encoding complexity is $n \text{log} n$ and their successive cancellation (SC) decoding complexity is $O(n \text{log} n)$, where $n$ is the length of the code. However, at moderate block lengths, their performance does not compete with world's best known codes, which prevents them from being implemented in practice. Also, their error exponent decreases slowly as the bock length increases, where the error-decay rate of polar codes under successive cancellation decoding is asymptotically $O(2^{-n^{0.5-\eps}})$. In this paper, we aim at providing techniques to make polar codes more practical, by providing schemes that improve their finite length performance, while preserving their low decoding complexity.

Concatenating inner polar codes with outer linear codes (or other variations of concatenation like parallel concatenation) is a promising path towards making them more practical \cite{BJE,EP}. By carefully constructing such codes, the concatenated construction can inherit the low encoding and decoding complexities of the inner polar code, while having significantly improved error-rate performance, in comparison with the inner polar code.
The performance and decoding complexity of the concatenated code, will also depend on the outer code used, the concatenation scheme and the decoding algorithms used for decoding the component codes. We chose Reed-Solomon (RS) codes  as outer codes as they are maximal distance separable (MDS) codes, and hence have the largest bounded-distance error-correction capability at a specified code rate. RS codes also have excellent burst error-correction capability. 

 Recent investigations have shown the possibility of improving the bound on the error-decay rate of polar codes by concatenating them with RS codes \cite{BJE}. However, this work assumed a conventional method of concatenation, which required the cardinality of the outer RS code alphabet to be exponential in the block length of the inner polar code, which makes it infeasible for implementation in practical systems.  In this paper, we propose a scheme for improving the error-decay rate of polar codes by concatenating them with interleaved block codes. When deploying our proposed scheme with outer interleaved RS block codes, the RS alphabet cardinality is no longer exponential and, in fact, is a design parameter which can be chosen arbitrarily. Furthermore, we show that the code parameters can be set such that the total scheme still achieves the capacity while the error-decay rate is asymptotically $2^{-N^{1-\eps}}$ for any $\eps > 0$, where $N$ is the total block length of the concatenated scheme. This bound provides considerable improvements upon $2^{-N^{0.5 - \eps}}$, the error-decay rate of Arikan's polar codes, and upon the bound of \cite{BJE}.

To construct the concatenated polar code at finite block length, we propose a rate-adaptive method to minimize the rate-loss resulting from the outer block code. It is known from the theory of polar codes that not all of the selected good bit-channels have the same performance. Some of the information bits observe very strong and almost noiseless channels, while some other information bits observe weaker channels. This implies that an unequal protection by the outer code is needed, i.e. the  strongest information bit-channels do not need another level of protection by the outer code, while the rest are protected by certain codes whose rates are determined by the error probability of the corresponding bit-channels. Hence, a criterion is established for determining the proper rates of the outer interleaved block codes.

We propose a successive method for decoding the RS-polar concatenated scheme, which is possible by the proposed interleaved concatenation scheme, where the symbols of each RS codeword are distributed over the same coordinates of multiple polar codewords. In the successive cancellation (SC) decoding of the inner polar codes, the very first bits of each polar codeword that are protected by the first outer RS code are decoded first. Then these bits are passed as symbols to the first inner decoder. RS decoding is done on the first RS word to correct any residual errors from  the polar decoders, and pass the updated information back to the SC polar decoders. Then the SC decoders of all polar codes update their first decoded bits, and use that updated information to continue successive decoding for the following bits corresponding to the symbols of the subsequent outer words.  Therefore, the errors from SC decoding do not propagate through the whole polar codeword, which  significantly improves the performance of our scheme. Another main advantage of this proposed scheme, is that all polar codes are decoded in parallel which significantly reduces the decoding latency.

Depending on the chosen outer code and its chosen decoding algorithm, the information exchanged between the inner SC decoders and the outer decoder can be soft information, as log-likelihood ratios (LLRs), or hard decisioned bits. We take advantage of the soft information generated by the successive cancellation decoder of polar codes to perform generalized minimum distance (GMD) list decoding \cite{F,K} for the outer RS code, which enhances the error performance. For further improvements, the SC decoder is modified so that it generates the likelihoods of all the possible RS symbols for GMD decoding. After GMD decoding of each component RS word, the most likely candidate codeword relative to the received word is picked from the list, and SC decoder utilizes the updated RS soft and hard outputs in further decoding of the component polar codes. In the case that outer codes are RS codes, more complex and better soft decoding algorithms exist, e.g. \cite{KV,EM}, however GMD was chosen to preserve the bound on the decoding complexity of the polar codes.

The rest of this paper is organized as follows: In Section\,\ref{sec:three}, the proposed scheme is explained in more details and we prove the bounds on the error-decay rate. We also explain how to modify the scheme to get a rate-adaptive construction, which significantly improves the rate of finite length constructions. In Section\,\ref{sec:four}, we describe our proposed decoding algorithms for the concatenated polar code scheme. Simulated results are shown in \ref{sec:fourb}. We conclude the paper by mentioning some directions for future work in Section\,\ref{sec:five}.

\section{RS-polar Concatenated Scheme}
\label{sec:three}

In this section, we describe our proposed scheme for concatenating polar codes with outer block codes. We consider the case when the outer code is an RS code. We establish bounds on the error correction performance and the decoding complexity of the proposed concatenated RS-polar code.

\begin{figure}
\centering
\includegraphics[width=3.2in]{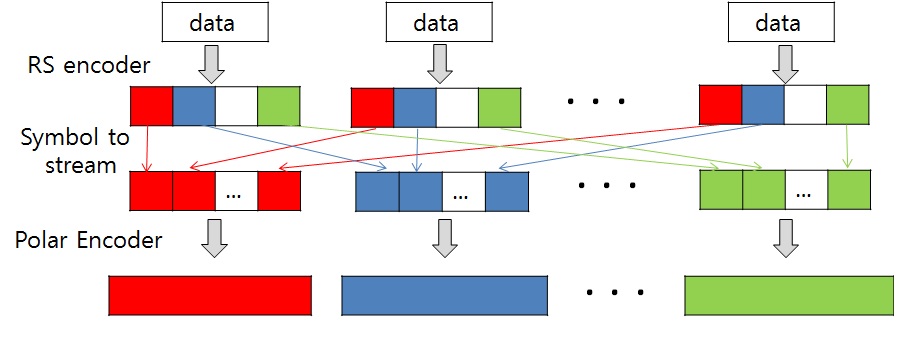}\\
\caption{Proposed concatenation scheme of Polar codes with outer interleaved block codes.
\label{fig:scheme}}
\vspace{-0.5cm}
\end{figure}

\subsection{Proposed RS-polar scheme with interleaving}
The proposed scheme for concatenating polar codes with outer RS codes is illustrated in Fig. \ref{fig:scheme}.
The symbols generated by a certain number of RS codewords are interleaved and converted into binary streams using a fixed basis to provide the input of the polar encoders. In a specific construction, the bits corresponding to the first symbols of all RS codewords are encoded into one polar codeword. Similarly, the information bits of the second polar codeword constitutes of all bit corresponding to the second symbol coordinates of all outer RS codewords. Hence, polar encoding can be done in parallel, which reduces the encoding latency.

More precisely, let $n$ and $m$ denote the lengths of the inner polar code and outer RS code, respectively. Let $k$ denote the number of input bits to each polar encoder i.e. the rate of each inner polar code is $R_I = k/n$. The symbols of the outer RS codes are drawn from the finite field $\F_{2^t}$, with cardinality $2^t$. It is assumed that $k$ is divisible by $t$. Hence, in the proposed scheme, the number of outer RS codes is $r=k/t$ and the number of inner polar codes is $m$. The rates of the outer RS codes will be specified later. Assume that $r$ RS codewords of length $m$ over $\F_{2^t}$ are given. For $i=1,2,\dots,r$, let $(c_{i,1},c_{i,2},\dots,c_{i,m})$ denote the $i$-th codeword.  For $j=1,2,\dots,m$, the $j$-th polar codeword is the output of the polar encoder of rate $k/n$ with the input $\bigl(\mathcal{I}(c_{1,j}),\mathcal{I}(c_{2,j}),\dots,\mathcal{I}(c_{r,j})\bigr)$, where $\mathcal{I}(c)$ maps a symbol $c \in \F_{2^t}$ to its binary image with $t$ bits. Hence, the total length of the concatenated word is $N=nm$.
 The interleaver proposed here between the inner and outer codes can be viewed as a structured block interleaver. Other polynomial interleavers may be considered for further improvements. The interleaver plays an important rule in this proposed scheme, as it helps to get rid of the huge field size required for the outer RS code by the previous scheme of \cite{BJE}.

\subsection{Asymptotic analysis of error correction performance}

Arikan and Telatar prove that the probability of block error of polar codes of length $N$ under SC decoding is bounded by $2^{-{N^{0.5-\eps}}}$, when $N$ is large enough \cite{AT}.
In our construction, assume all outer RS codes have the same rate $R_o$.
In \Lref{lemma1}, it is shown that the error probability of the concatenated code is bounded by $2^{-(n^{0.5-\eps}(1-R_o)/2 - 1)m}$. Then in \Tref{theorem1}, we prove that $m$, $n$ and $R_o$ can be set in such a way that the error probability of the concatenated code scales as $2^{-N^{1-\eps}}$, for any $\eps > 0$, asymptotically, while the concatenated code is still capacity achieving. This significantly improves the error-decay rate compared to polar codes with the same length $N$.

\begin{lemma}
\label{lemma1}
In the proposed RS-polar concatenated scheme, for any $\eps > 0$ and large enough $n$, the probability of frame error is upper bounded by $ 2^{-(n^{0.5-\eps}(1-R_o)/2 - 1)m}$.
\end{lemma}
\begin{proof}
Assuming a bounded-distance RS decoder, the error correction capability of RS codes is $\tau = \left\lfloor (1-R_o)m/2\right\rfloor$.
Let $\mathcal{E}$ be the frame error probability (FEP) of the concatenated code, then it can be shown that
$\mathcal{E} \leq {m \choose \tau + 1} P_e^{\tau +1}$,
where $P_e$ is the codeword error probability of the inner polar code.
Hence $\mathcal{E} \leq {m \choose \tau +1} 2^{-n^{0.5 - \eps}(\tau +1)}$
and the Lemma follows.
\end{proof}
\begin{theorem}
\label{theorem1}
For any $\eps > 0$, the lengths of the inner polar code and outer RS code, and the rate of outer RS code can be set
such that the frame error probability of the concatenated code of total length $N$ is asymptotically upper bounded by $2^{-N^{1-\eps}}$, while the scheme is still capacity-achieving.
\end{theorem}
\begin{proof}
The length of the inner polar code $n$, the length of the outer RS code $m$, and the rate of outer RS code $R_o$
can be set as follows:
\begin{equation} \nonumber
n = N^{\eps},\ m = N^{1-\eps}\ \text{, and}\ R_o = 1 - 4N^{-\eps(0.5-\eps)}.
\end{equation}
Substituting $n$, $m$, and $R_o$ into the bound given by \Lref{lemma1}, one gets
\begin{equation}
\mathcal{E} \leq 2^{-(n^{0.5 - \eps}(1-R_o)/2 - 1)m} = 2^{-N^{1-\eps}}
\end{equation}
as the upper bound on the FEP.
With above settings, $R_o \rightarrow 1$, as $N \rightarrow \infty$. Hence, the
 rate of the concatenated polar code also approaches the capacity, since the inner polar code is proven to be capacity-achieving.
\end{proof}

\subsection{Decoding complexity}

To compute the decoding complexity of the concatenated RS-polar code, we take into account the decoding complexity of both the inner polar code and the outer RS code.

The decoding complexity of the polar code using successive cancellation decoding is given by $O(n\log n)$, with $n$ being the length of the polar code \cite{Arikan}. Since there are $m$ inner polar codes in the proposed concatenated scheme, the total complexity of decoding the inner polar codes is $O(nm\log n)$, which is bounded by $O(N \log N)$.

A well-known hard-decision bounded-distance RS decoding method is the Berlekamp-Massey (BM) algorithm. The BM algorithm is a syndrome-based method which finds the error locations and error magnitudes separately. The decoding complexity is known to be $O(m^2)$ operations over the field $\F_{2^t}$. One main advantage of the proposed concatenated code is that the RS code alphabet size scales linearly with the desired RS code length $m$, whereas it is exponential in terms of $n$ in the previous scheme \cite{BJE}. Gao proposed a syndrome-less RS decoding algorithm  that uses fast Fourier transform and computes the message symbols directly without computing error locations or error magnitudes \cite{G}. For RS codes over arbitrary fields, the asymptotic complexity of syndrome-less decoding based on multiplicative FFT techniques was shown to be $O(m \log^2 m \log\log m)$. Hence by deploying syndrome-less RS decoding, the total complexity of decoding the outer RS codes can be at most $O(nm \log^2 m \log\log m)$ which is bounded by $O(N \log^2 N  \log\log N )$.

Therefore, the total decoding complexity of the proposed concatenated RS-polar code can be asymptotically  bounded by $O(N \log^2 N  \log\log N )$.

\subsection{Rate-adaptive construction of RS-polar concatenated scheme}
\label{rate-adaptive}
Due to the polarization phenomenon of polar codes, not all bit-channels chosen to carry the information bits have the same reliability. An outer RS code is not actually needed for the strongest bit-channels of the inner polar code, since the corresponding information bits are already well-protected.
Since in our proposed concatenation scheme, each RS codeword is re-encoded by same bit-channel indices across the different polar codewords, the rate of each RS code can be properly assigned to protect the polarized bit-channels in such a way that all the information bits are almost equally protected. Suppose that the probability of error for each of the input bits to the polar code is given. Let $k$ be the information block length of the inner polar code. For $i = 1,2,\dots,k$, let $P_{i}$ be the probability that an error occurs when decoding the $i$-th information bit with the SC decoder, assuming that all the first $i-1$ information bits were successfully recovered. Suppose that the outer RS code is over $\F_{2^t}$, for some integer $t$. Then the total number of RS codes is $k/t$. Then, the first $k/t$ information bits of each polar codeword form one symbol for the first RS code, the next $k/t$ information bits form one symbol for the second RS code, etc. If we consider one of the inner polar codes, the probability that the first RS symbol has an error at the ouptut of the SC polar decoder is given by $1-(1-P_1)(1-P_2)\dots(1-P_t)$. In general, for $i=1,2,\dots,k/t$, the probability that the $i$-th RS symbol is in error assuming that all the previous symbols were decoded successfully is given by $Q_i = 1-(1-P_{it-t+1})(1-P_{it-t+2})\dots(1-P_{it})$.

The design criterion is as follows. Let $\Pef$ be the target  FEP of the concatenated code. Then, for $i=1,2,..,k/t$, let $\tau_i$ be the smallest positive integer such that
\be{redundancy_assign}
{m \choose \tau_i+1} Q_i^{\tau_i+1} < \frac{t\Pef}{k}.
\ee
Then, the proposed rate-adaptive concatenation scheme deploys a $\tau_i$-error correcting RS code for the $i$-th outer RS code. The following lemma shows that the FEP $\Pef$ is guaranteed.
\begin{lemma}
Suppose that the $i$-th outer RS code is a $(m,m-2\tau_i)$ code, for $i=1,2,\dots,k/t$, where $\tau_i$ is determined by \eq{redundancy_assign}. Then, the total frame error probability for the RS-polar concatenated code is less than $\Pef$.
\end{lemma}

The proof follows from the design criterion \eq{redundancy_assign} and is omitted due to space limitations.

The rate-adaptive design criterion, described above, requires knowledge of the individual bit-channel error probability. Whereas it can be calculated for erasure channels, we take a numerical approach to solve this problem for an arbitrary channel, e.g. additive white Gaussian noise (AWGN) channel: Assume that all previous input bits $1,2,\dots,i-1$ are provided to the SC decoder by a genie when the $i$-th bit is decoded. For bit $i = 1,2\dots,n$, the decoder is run for a sufficiently large number of independent inputs to get an estimate of the probability of the event that the $i$-th bit is not successfully decoded, given that the bits indexed by $1,2,\dots,i-1$ were successfully decoded. An alternative way is to use the method introduced in \cite{TV} which provides tight upper and lower bounds on the bit-channel error probability.

\section{Decoding methods for the RS-polar Concatenated Scheme}
\label{sec:four}

Conventional decoding of serially concatenated codes is done by decoding the received data with the inner decoder, whose output is decoded by the outer decoder. The output of the outer decoder is actually the decoded data. However, this straightforward way of decoding does not show a good performance for our RS-polar concatenated scheme in short block lengths. Therefore, as the first step in improving the performance of our proposed concatenated scheme, we take advantage of the successive way of decoding polar codes to propose a successive decoding method for the serially concatenated RS-polar code. The proposed decoding scheme, SC polar decoding of the different polar codes is done in parallel with small decoding latency.

\subsection{Successive decoding}

The main disadvantage of the successive cancellation decoding of polar codes is that once an error occurs, it may propagate through the whole polar codeword. Since the information block of the polar code is protected with an outer RS code, any errors in the decoded bits can be corrected using the outer code while the SC decoder evolves. This can potentially mitigate the error propagation problem and consequently results in improvement in the FEP of the proposed scheme at finite block lengths. Hence, we propose the following successive decoding algorithm for our proposed serially concatenated RS-polar code:
\begin{itemize}
\item
Start by SC decoding of the very first bits of each polar codeword corresponding to the first RS symbol, that is if the outer RS code is over $\F_{2^t}$, decode the first $t$ bits of each polar codeword. This operation can be done in parallel for all the inner polar codewords.
\item
Pass the $mt$ hard-decisioned output bits as $m$ symbols to form the first RS word, and decode it with the bounded-distance RS decoder.
\item
Update the first $t$ decoded bits of all $m$ polar codewords using the RS decoder output, and use them to continue SC decoding.
\item
Keep doing this for the second and third RS codewords, etc.
\end{itemize}

\subsection{Generalized minimum distance decoding for the outer code}

Generalized minimum distance (GMD) decoding was introduced by Forney in \cite{F}, where the soft information is used with  algebraic bounded-distance decoding to generate a list of codewords. In the concatenated code, the likelihood of each symbol can be computed given the LLRs of the corresponding bits generated by the SC decoder of the inner polar code. The $m$ symbols of each RS word are sorted with respect to their likelihoods.  The $\alpha$ least likely symbols are declared as erasures, where the case of $\alpha=0$ is the same as a regular RS decoding. In conventional GMD decoding, errors and erasures decoding of RS codes is run for $\alpha = \{0,2,4, \dots, d-1\}$, where $d$ is the minimum distance of the RS code. This gives a list of size at most $(d+1)/2$ candidate RS codewords at the output of the decoder. The decoder picks the closest one to the received word, which is then passed to the polar decoders. A naive way of implementing the GMD decoding increases the complexity by a factor of $(d+1)/2$. However, Koetter derived a fast GMD decoding algorithm which removes this factor \cite{K}, and hence GMD can be deployed in decoding our RS-polar code while preserving our decoding complexity bound.
Since the SC decoder actually computes the LLR's of the bits in each symbol, the likelihood of each symbol that is passed to RS decoder can be computed. The symbol likelihoods from the different polar SC decoders are used for GMD decoding of each RS code.
\subsection{GMD-ML decoding for the outer code}

In the previous subsection, at the last step of GMD RS decoding, the candidate in the generated list of codewords that is the closest one to the received word is picked. Here, we further improve the performance by actually picking the most likely codeword based on soft information from the polar decoder. There are two approaches to utilize this idea. The first approach is to approximate the symbol probabilities using the bit LLR's generated by the polar decoder. This is not precise, since the bit LLR's in each symbol are not independent. However, it gives about $0.1$dB gain with no cost in complexity. We pick the best codeword from the list generated by GMD decoder based on its estimated probability given by the product of estimated symbol probabilities. We call this approach \emph{GMD with approximate ML decoding}. In the second approach, we modify the SC decoder of polar code to output the soft information for all the possible symbols. As pointed out before, for a symbol constituting of $t$ bits, the LLR of each bit depends on the previous bits in the symbol. In order to compute the exact symbol probabilities, the SC polar decoder computes the probabilities of all the $2^t$ symbols by traversing all the possible $2^t$ paths, for each consecutive $t$ bits. This increases the complexity of polar decoder by a constant factor of $\sum_{i=0}^{t-1}2^i/t$. This enables GMD with exact ML decoding. Also, since the SC decoder recursively calculates the LLRs, the LLRs computed along each path are saved so when the correct symbol is picked by the outer decoder, the LLRs computed along the corresponding path are picked to proceed with SC decoding for the next $t$ bits.

\section{Simulated Numerical Performance \label{sec:fourb}}
Transmission over AWGN channel is assumed. The inner polar code of length $2^9 = 512$ and outer Reed-Solomon code of length $15$ over $\F_{2^4}$ are considered. The rates of inner and outer codes are designed such that the total rate of scheme is $1/3$.  We use the method explained in \ref{rate-adaptive} to construct the concatenated code. The actual probability of error of the bit-channels under SC decoding corresponding to a polar code of length $512$ are estimated over an AWGN channel with $2$ dB SNR. The size of sample space is $10^5$. The design criterion for the outer RS code is as follows:
\begin{itemize}
\item
Fix $k$, the input length of the inner polar code ($k$ has to be a multiple of $4$).
\item
Then pick the best $k$ bit-channels that have smaller probability of errors and sort them with respect to their index. Suppose that we get $(i_1,i_2,\dots,i_k)$ with corresponding bit-channel probability of errors $(p_1,p_2,\dots,p_k)$.
\item
Fix a target probability of error $P_e$ for each of the small sub-blocks of length $4$. ($2^4$ is the size of the alphabet for RS code)
\item
For $j=1,2,\dots,k/4$, let $Q_j= p_{4j-3}+ p_{4j-2}+ p_{4j-1}+ p_{4j}$ be the union bound on the probability of error of the $j$-th sub-block of length 4. Then pick the smallest integer $\tau_j$ such that
$$
\sum_{l=\tau_j+1}^m {m \choose l}  Q_j^l (1-Q_j)^{m-l} < P_e
$$
where $m$ is the length of the RS code.
\item
Compute the total rate of the scheme to see of it is above $1/3$ or below $1/3$ (The target rate is $1/3$). Then set a new target probability of error $P_e$ accordingly and repeat these steps.
\end{itemize}
We also let $k$ to vary between $170$ (inner rate $1/3$) and $256$ (inner rate $1/2$). At the end $k = 204$ is picked which results in the best performance. 

The performance with the proposed construction and decoding techniques is shown in Figure\,\ref{plot4}. The results are compared with a polar code of the same length $512$. Since we can decode all the $15$ sub-blocks of the polar code in parallel, the two schemes have the same decoding latency.  For the concatenation scheme, we define the block error rate as the error rate of the sub-blocks of data corresponding to each inner polar codeword. The aim is to have a fair comparison with a polar code of the same block size $512$ when the rates are equal, where the rate loss due to RS outer code is taken into account.
 At an error probability of $10^{-4}$, it is observed that the proposed concatenated RS-polar code has more than 1.8 dB SNR gain over the non-concatenated polar code with the same decoding latency. Also, the proposed GMD-ML successive decoding algorithm offers more than 2 dB SNR gain over conventional serial decoding of the RS-polar concatenated code, in which the outer RS code is a $(15,11)$ code and the inner polar code is a $(512,232)$ code.

\begin{figure}[t]
\centering
%

\begin{tikzpicture}

\definecolor{mycolor1}{rgb}{1,0,1}

\begin{semilogyaxis}[
scale only axis,
width = 2.75in,
height = 2.0625in,
xmin=1, xmax=6,
ymin=1e-06, ymax=1,
xlabel={{\footnotesize $E_b/N_0$ [dB]}},
ylabel={{\footnotesize Block error rate}},
title = {{\footnotesize $n = 512$, rate = $1/3$}},
xmajorgrids,
ymajorgrids,
yminorgrids,
legend entries={{\scriptsize polar code},{\scriptsize regular RS-polar},{\scriptsize GMD-ML decoding}},
legend style={nodes=right}]

\addplot [
color=blue,
solid,
mark=x,
mark options={solid}
]
coordinates{
 (1,0.3226)
 (1.5,0.1167)
 (2,0.0458)
 (2.5,0.0092)
 (3,0.0018)
 (3.5,0.00029032)
 (4,3.6863e-05)

};

\addplot [
color=mycolor1,
solid,
mark=o,
mark options={solid}
]
coordinates{
 (2.5,0.9804)
 (3,0.5701)
 (3.5,0.072)
 (4,0.0029)
 (4.5,0.00011111)

};

\addplot [
color=black,
solid,
mark=square,
mark options={solid}
]
coordinates{
 (1,0.5138)
 (1.5,0.0355)
 (1.75,0.0073)
 (2,0.0011)
 (2.25,7.7348e-05)

};

\end{semilogyaxis}

\end{tikzpicture}
\caption{Performance of the concatenated scheme using GMD-ML decoding technique}
\label{plot4}
\vspace{-0.5cm}
\end{figure}
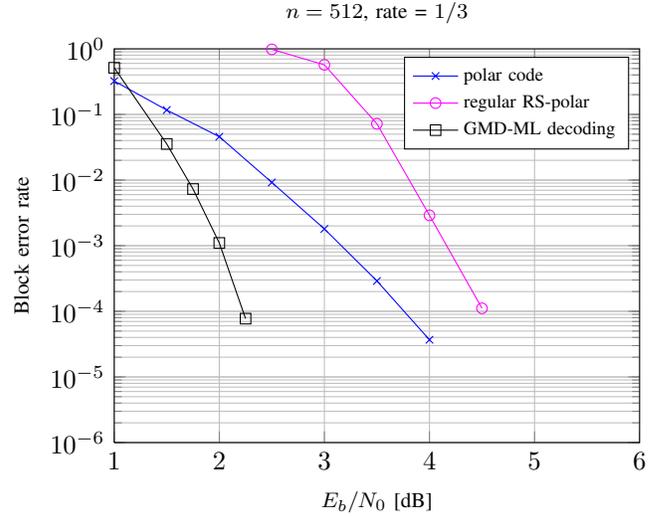

\section{Conclusions and Future Work}
\label{sec:five}

In this paper, we proved that by carefully concatenating the recently invented polar codes with Reed-Solomon codes, a significant improvement in the error-decay rate compared to non-concatenated polar codes is possible. The parameters of the scheme can be set to inherit the capacity-achieving property of polar codes while working in the same regime of low complexity. We developed several construction methods and decoding techniques to improve the performance at finite block lengths, which is a step to making polar codes more practical. There are some directions for future work as follows. The methods described in this paper can be in general applied to concatenation of polar codes with non-binary block codes. The construction and decoding methods can also be used when the outer code is a binary code by grouping each $t$ bits, where $t$ is a complexity parameter to be optimized. In such case, the outer codes should be selected to have low complexity soft decoding algorithms. Further improved decoding algorithms for polar codes as CRC-aided list decoding \cite{TV2} have the potential to improve the performance of the proposed successive decoding of concatenated polar codes.

\end{document}